\documentclass[a4paper,11pt]{amsart}

\usepackage[margin=1in,includehead,includefoot]{geometry}
\usepackage[OT4]{fontenc}
\usepackage{amsthm,amsmath,amssymb}
\usepackage{enumitem}
\usepackage{tikz}
\usepackage{graphicx,psfrag}
\usepackage[pdftitle={Outerplanar graph drawings with few slopes},
            pdfauthor={Kolja Knauer, Piotr Micek, Bartosz Walczak}]{hyperref}

\newtheorem*{mainthm}{Main Theorem}
\newtheorem{theorem}{Theorem}
\newtheorem{lemma}[theorem]{Lemma}

\theoremstyle{remark}
\newtheorem{case}{Case}
\newtheorem{subcase}{Subcase}[case]

\newenvironment{enumeratedlemma}{\begin{enumorig}[label=\textup{(\arabic{theorem}.\arabic*)}, itemsep=1.5mm plus 1.5mm, topsep=1.5mm plus 1.5mm, leftmargin=*]}{\end{enumorig}}

\renewenvironment{itemize}{\begin{itemorig}[label=\textbullet, noitemsep, topsep=1.5mm plus 1.5mm, labelsep=.6em, labelindent=.4em, leftmargin=*]}{\end{itemorig}}

\newcounter{savedenumi}
\def\enumeratedlemmaintertext#1{
    \setcounter{savedenumi}{\value{enumi}}
    \endenumeratedlemma #1\enumeratedlemma
    \setcounter{enumi}{\value{savedenumi}}
}

\newcounter{enumeratedlemmaproofitem}

\newenvironment{enumeratedlemmaproofitem}[1]{
	\begingroup
	\def\pushQED##1{}
	
	\setcounter{case}{0}
	\stepcounter{enumeratedlemmaproofitem}
	\begin{proof}[Proof of \ref{#1}]
}{\end{proof}\endgroup}

\newbox\itemcasebox
\newdimen\itemcasedim
\def\itemcase#1#2{
    \setbox\itemcasebox\hbox{#1}
    \ifdim\itemcasedim<\wd\itemcasebox\itemcasedim\wd\itemcasebox\fi
    \expandafter\def\expandafter\itemizedcasesbody\expandafter
      {\itemizedcasesbody\item\leavevmode\hbox to\itemcasedim{#1\hfil}\quad\enspace #2}
}
\newenvironment{itemizedcases}{
    \itemcasedim 0pt
    \def\itemizedcasesbody{}
}{\begin{itemize}\itemizedcasesbody\end{itemize}}

\def\ff{\mathbf{f}}
\def\Raysign{\mathit{R}}
\def\LBRsign{\mathit{LB}}
\def\RBRsign{\mathit{RB}}
\def\BRsign{\mathit{B}}
\def\BRUsign{\mathit{\bar{B}}}
\def\Ray#1#2{\Raysign(#1;#2)}
\def\LBR#1#2{\LBRsign(#1;#2)}
\def\RBR#1#2{\RBRsign(#1;#2)}
\def\BR#1#2#3{\BRsign(#1;#2,#3)}
\def\BRU#1#2#3#4{\BRUsign(#1;#2,#3;#4)}

\def\abs#1{\lvert#1\rvert}

\let\le\leqslant
\let\ge\geqslant
\let\epsilon\varepsilon
\let\subset\subseteq

\def\set#1{\{#1\}}

\allowdisplaybreaks[1]

\makeatletter
\let\old@setaddresses\@setaddresses
\def\@setaddresses{\bgroup\parindent 0pt\let\scshape\relax\old@setaddresses\egroup}
\makeatother

\title{Outerplanar graph drawings with few slopes}
\author{Kolja Knauer}
\author{Piotr Micek}
\author{Bartosz Walczak}

\address[Kolja Knauer]{Laboratoire d'Informatique, de Robotique et de Micro\'electronique de Montpellier, France}
\email{kolja.knauer@gmail.com}

\address[Piotr Micek, Bartosz Walczak]{Theoretical Computer Science Department, Faculty of Mathematics and Computer Science, Jagiellonian University, Krak{\'o}w, Poland}
\email{micek@tcs.uj.edu.pl, walczak@tcs.uj.edu.pl}

\thanks{Journal version of this paper appeared in \emph{Comput.\ Geom.}, 47(5):614--624, 2014.}
\thanks{Preliminary version of this paper appeared in Joachim Gudmundsson, Juli\'an Mestre, and Taso Viglas, editors, \emph{Computing and Combinatorics (COCOON 2012)}, volume 7434 of \emph{Lecture Notes Comput.\ Sci.}, pages 323--334.\ Springer, Berlin, 2012.}

\thanks{Kolja Knauer was supported by DFG grant FE-340/8-1 under ESF EuroGIGA project GraDR and by ANR TEOMATRO grant ANR-10-BLAN 0207.}
\thanks{Piotr Micek was supported by Ministry of Science and Higher Education of Poland grant 884/N-ESF-EuroGIGA/10/2011/0 under ESF EuroGIGA project GraDR\@.}
\thanks{Bartosz Walczak was supported by Ministry of Science and Higher Education of Poland grant 884/N-ESF-EuroGIGA/10/2011/0 under ESF EuroGIGA project GraDR and by Swiss National Science Foundation grant 200020-144531.}

\begin{document}
\baselineskip 14pt

\begin{abstract}
We consider straight-line outerplanar drawings of outerplanar graphs in which a small number of distinct edge slopes are used, that is, the segments representing edges are parallel to a small number of directions. 
We prove that $\Delta-1$ edge slopes suffice for every outerplanar graph with maximum degree $\Delta\ge 4$. 
This improves on the previous bound of $O(\Delta^5)$, which was shown for planar partial $3$-trees, a superclass of outerplanar graphs. 
The bound is tight: for every $\Delta\ge 4$ there is an outerplanar graph with maximum degree $\Delta$ that requires at least $\Delta-1$ distinct edge slopes in an outerplanar straight-line drawing.
\end{abstract}

\maketitle

\section{Introduction}

A \emph{straight-line drawing} of a graph $G$ is a mapping of the vertices of $G$ into distinct points of the plane and of the edges of $G$ into straight-line segments connecting the points representing their end-vertices and passing through no other points representing vertices. 
If it leads to no confusion, in notation and terminology, we make no distinction between a vertex and the corresponding point, and between an edge and the corresponding segment. 
The \emph{slope} of an edge in a straight-line drawing is the family of all straight lines parallel to this edge. 
The \emph{slope number} of a graph $G$, a parameter introduced by Wade and Chu \cite{Wad-94}, is the smallest number $s$ such that there is a straight-line drawing of $G$ using $s$ slopes.

Since at most two edges at each vertex can use the same slope, $\lceil\frac{\Delta}{2}\rceil$ is a lower bound on the slope number of a graph with maximum degree $\Delta$. 
Dujmovi{\'c} and Wood \cite{Duj-04} asked whether the slope number can be bounded from above by a function of the maximum degree. 
This has been answered independently by Bar{\'a}t, Matou{\v{s}}ek and Wood \cite{Bar-06}, Pach and P{\'a}lv{\"o}lgyi \cite{Pac-06}, and Dujmovi{\'c}, Suderman and Wood \cite{Duj-07a} in the negative: graphs with maximum degree $5$ can have arbitrarily large slope number. 
On the other hand, Mukkamala and P{\'a}lv{\"o}lgyi \cite{Muk-12} proved that graphs with maximum degree $3$ have slope number at most $4$, improving earlier results of Keszegh, Pach, P{\'a}lv{\"o}lgyi and T{\'o}th \cite{Kes-08} and of Mukkamala and Szegedy \cite{Muk-09}. 
The question whether the slope number of graphs with maximum degree $4$ is bounded by a constant remains open. 

The situation looks different for \emph{planar} straight-line drawings, that is, straight-line drawings in which no two edges intersect in a point other than a common endpoint. 
It is well known that every planar graph admits a planar straight-line drawing \cite{Far-48,Koe-36,Wag-36}. 
The \emph{planar slope number} of a planar graph $G$ is the smallest number $s$ such that there is a planar straight-line drawing of $G$ using $s$ slopes. 
This parameter was first studied by Dujmovi{\'c}, Eppstein, Suderman and Wood \cite{Duj-07b} in relation to the number of vertices. 
They also asked whether the planar slope number of planar graphs is bounded in terms of the maximum degree. 
Jel{\'{\i}}nek, Jel{\'{\i}}nkov{\'a}, Kratochv{\'{\i}}l, Lidick{\'y}, Tesa{\v{r}} and Vysko{\v{c}}il \cite{Jel-13} gave an upper bound of $O(\Delta^5)$ for planar graphs of treewidth at most $3$. 
The problem has been solved in full generality by Keszegh, Pach and P{\'a}lv{\"o}lgyi \cite{Kes-13}, who showed (with a non-constructive proof) that the planar slope number is bounded from above by an exponential function of the maximum degree. 
It is still an open problem whether this can be improved to a polynomial upper bound. 

In the present paper, we consider drawings of outerplanar graphs. 
The above-mentioned result of Jel{\'{\i}}nek et al.\ implies that outerplanar graphs admit planar drawings with $O(\Delta^5)$ slopes, as they have treewidth at most $2$. 
A straight-line drawing of a graph $G$ is \emph{outerplanar} if it is planar and all vertices of $G$ lie on the outer face. 
The \emph{outerplanar slope number} of an outerplanar graph $G$ is the smallest number $s$ such that there is an outerplanar straight-line drawing of $G$ using $s$ slopes. 
It is proved in \cite{Duj-07b} that the outerplanar slope number of any outerplanar graph is at most the number of its vertices. 
We provide a tight bound on the outerplanar slope number in terms of the maximum degree. 

\begin{mainthm}
For\/ $\Delta\ge 4$, every outerplanar graph with maximum degree at most\/ $\Delta$ has outerplanar slope number at most\/ $\Delta-1$.
\end{mainthm}

That the bound of $\Delta-1$ is tight is witnessed by a graph consisting of a cycle $C$ with $2\Delta-3$ vertices $v_1,\ldots,v_{2\Delta-3}$ each of which has $\Delta-2$ additional private neighbors. 
In any outerplanar straight-line drawing of this graph with $\Delta-2$ edge slopes, $C$ must be the boundary of an inner face. 
It cannot be strictly convex, as in a strictly convex polygon each slope can be used by at most two edges. 
Therefore, some angle of this face, say at $v_i$, is not strictly convex. 
Each of the private neighbors of $v_i$ needs to be connected with $v_i$ by an edge lying outside the cycle. 
This is a contradiction, because at most $\Delta-3$ slopes are available for such edges. 
Moreover, for $\Delta\in\set{2,3}$, the lower bound is $3$ as witnessed by the triangle. 

The tight bounds for the outerplanar slope number of outerplanar graphs with maximum degree $\Delta\in\set{1,2,3}$ are also determined. 
It is $1$ for $\Delta=1$ and $3$ for $\Delta\in\set{2,3}$. 
For the latter, the upper bound follows from the Main Theorem applied with $\Delta=4$, while the tightness is witnessed by a triangle. 

The proof of our theorem is constructive and yields an algorithm to produce a claimed drawing that performs a linear number of arithmetic operations on rationals.

\section{Basic definitions}

For the remainder of the paper, we assume that an outerplanar drawing of a graph $G$ with maximum degree at most $\Delta$ is given, where $\Delta\ge 4$. 
This drawing determines the cyclic ordering of edges around every vertex. 
We produce an outerplanar straight-line drawing of $G$ with few edge slopes which preserves this ordering at every vertex. 
The set of slopes that we use depends only on $\Delta$, so we can draw each connected component of $G$ separately. 
Therefore, for the remainder of the paper, we assume that $G$ is connected. 
Our construction is inductive---it composes the entire drawing of $G$ from drawings of subgraphs of $G$ that we call bubbles. 

We distinguish the \emph{outer face} of $G$ (the one that is unbounded in the given drawing of $G$ and contains all vertices on the boundary) from the \emph{inner faces}. 
The edges on the boundary of the former are \emph{outer edges}, while all remaining ones are \emph{inner edges}. 
A \emph{snip} is a simple closed counterclockwise-oriented curve $\gamma$ which
\begin{itemize}
\item passes through some pair of vertices $u$ and $v$ of $G$ (possibly being the same vertex) and through no other vertex of $G$,
\item on the way from $v$ to $u$ goes entirely through the outer face of $G$ and crosses no edge of $G$,
\item on the way from $u$ to $v$ (considered only if $u\ne v$) goes through inner faces of $G$ possibly crossing some inner edges of $G$ that are not incident to $u$ or $v$, each at most once,
\item crosses no edge of $G$ incident to $u$ or $v$ at a point other than $u$ or $v$.
\end{itemize}
Every snip $\gamma$ defines a \emph{bubble} $H$ in $G$ as the subgraph of $G$ induced on the vertices lying on or inside $\gamma$. 
Note that $H$ is a connected induced subgraph of $G$ as $\gamma$ crosses no outer edges. 
The \emph{roots} of $H$ are the vertices $u$ and $v$ together with all vertices of $H$ adjacent to $G-H$. 
The snip $\gamma$ breaks the cyclic clockwise order of the edges of $H$ around $u$ or $v$, making it a linear order, which we envision as going from left to right. 
In particular, we call the first edge in this order \emph{leftmost} and the last one \emph{rightmost}. 
Similar left-to-right orderings of edges are defined at the remaining roots of $H$, except that in their case the cyclic order is broken by the edges connecting $H$ to $G-H$. 
The \emph{root-path} of $H$ is the simple oriented path $P$ in $H$ that starts at $u$ with the rightmost edge, continues counterclockwise along the boundary of the outer face of $H$, and ends at $v$ with the leftmost edge. 
If $u=v$, then the root-path consists of that single vertex only. 
All roots of $H$ lie on the root-path---their sequence in the order along the root-path is the \emph{root-sequence} of $H$. 
A bubble with $k$ roots is called a \emph{$k$-bubble}. 
See Figure \ref{fig:bubble} for an illustration.

\begin{figure}[t]
\begin{center}
\psfrag{1}{{$H_1$}}
\psfrag{2}{{$H_2$}}
\psfrag{3}{{$H_3$}}
\psfrag{4}{{$H_4$}}
\psfrag{5}{{$H_5$}}
\psfrag{6}{{$H_6$}}
\psfrag{7}{{$H_7$}}
\psfrag{8}{{$H_8$}}
\psfrag{9}{{$H_9$}}
\psfrag{10}{{$H_{10}$}}
\psfrag{a}{$u$}
\psfrag{b}{$v$}
\psfrag{c}{$w$}
\psfrag{P}{}
\includegraphics[scale=1.25]{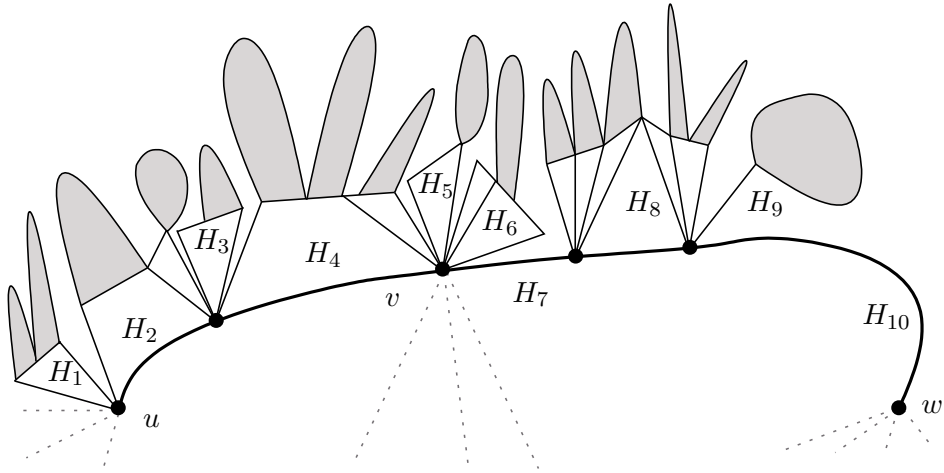}
\caption{A $3$-bubble $H$ with root-path drawn thick, root-sequence $u,v,w$ (connected to the remaining graph by dotted edges), and splitting sequence $(H_1,\ldots,H_{10})$, in which $H_1$, $H_3$, $H_5$, $H_6$, and $H_9$ are v-bubbles, while $H_2$, $H_4$, $H_7$, $H_8$, and $H_{10}$ are e-bubbles. 
There are six $2$-bubbles with roots $u$ and $v$ contained in $H$: $(H_1,H_2,H_3,H_4)$, $(H_1,H_2,H_3,H_4,H_5)$, $(H_1,H_2,H_3,H_4,H_5,H_6)$, $(H_2,H_3,H_4)$, $(H_2,H_3,H_4,H_5)$, and $(H_2,H_3,H_4,H_5,H_6)$.}
\label{fig:bubble}
\end{center}
\end{figure}

Except at the very end of the proof where we regard the entire $G$ as a bubble, we deal with bubbles $H$ whose first root $u$ and last root $v$ are adjacent to $G-H$. 
For such bubbles $H$, all the roots, the root-path, the root-sequence and the left-to-right order of edges at every root do not depend on the particular snip $\gamma$ used to define $H$. 
Specifically, for such bubbles $H$, the roots are exactly the vertices adjacent to $G-H$, while the root-path consists of the edges of $H$ incident to inner faces of $G$ that are contained in the outer face of $H$. 
From now on, we will refer to the roots, the root-path, the root-sequence and the left-to-right order of edges at every root of a bubble $H$ without specifying the snip $\gamma$ explicitly. 

Bubbles admit a natural decomposition, which is the base of our recursive drawing. 

\begin{lemma}\label{lem:split}
Let\/ $H$ be a bubble with root-path\/ $v_1\ldots v_k$. 
Every component of\/ $H-\set{v_1,\ldots,v_k}$ is adjacent to either one vertex among\/ $v_1,\ldots,v_k$ or two consecutive vertices from\/ $v_1,\ldots,v_k$. 
Moreover, there is at most one component adjacent to\/ $v_i$ and\/ $v_{i+1}$ for\/ $1\le i<k$.
\end{lemma}

\begin{proof}
Let $C$ be a connected component of $H-\set{v_1,\ldots,v_k}$.
As $H$ itself is connected, $C$ must be adjacent to a vertex from $v_1,\ldots,v_k$. 
In order to get a contradiction, suppose that $C$ is connected to two non-consecutive vertices $v_i$ and $v_j$. 
Let $P$ be a simple $v_i,v_j$-path having all internal vertices in $C$. 
Let $P'=v_i\ldots v_j$ be the subpath of the root-path of $H$ connecting $v_i$ and $v_j$. 
Since $v_1\ldots v_k$ is the root-path of $H$, all edges connecting the internal vertices of $P'$ to $G-H$ are inner edges. Hence, also the edges of $P'$ lie on inner faces which are not faces of $H$.
The symmetric difference of all these inner faces considered as sets of edges is a simple cycle containing $P'$ as a subpath. 
Let $P''$ denote the other $v_i,v_j$-subpath of that cycle. 
It is internally disjoint from $P$ and $P'$. 
Moreover, $P''$ and $P$ together enclose $P'$ and thus the internal vertices of $P'$ do not lie on the outer face---contradiction. 

Now, to prove the second statement, suppose that for some $i$ two components $C$ and $C'$ of $H-\set{v_1,\ldots,v_k}$ are adjacent to both $v_i$ and $v_{i+1}$. 
We find two internally disjoint $v_i,v_{i+1}$-paths $P$ and $P'$ through $C$ and $C'$, respectively. 
As in the paragraph above, we use the fact that $v_iv_{i+1}$ is contained in an inner face, which is not a face of $H$.
The third path $P''$ is obtained from that face by deleting the edge $v_iv_{i+1}$.
It follows that $P$, $P'$, and $P''$ form a subdivision of $K_{2,3}$, which contradicts the outerplanarity of $G$.
\end{proof}

Lemma~\ref{lem:split} allows us to assign each component of $H-\set{v_1,\ldots,v_k}$ to a vertex of $P$ or an edge of $P$ so that every edge is assigned at most one component. 
For a component $C$ assigned to a vertex $v_i$, the graph induced on $C\cup\set{v_i}$ is called a \emph{v-bubble}. 
Such a v-bubble is a $1$-bubble with root $v_i$. 
For a component $C$ assigned to an edge $v_iv_{i+1}$, the graph induced on $C\cup\set{v_i,v_{i+1}}$ is called an \emph{e-bubble}. 
Such an e-bubble is a $2$-bubble with roots $v_i$ and $v_{i+1}$. 
If no component is assigned to an edge of $P$, then we let that edge alone be a \emph{trivial e-bubble}. 
All v-bubbles of $v_i$ in $H$ are naturally ordered by their clockwise arrangement around $v_i$ in the drawing. 
All this leads to a decomposition of the bubble $H$ into a sequence $(H_1,\ldots,H_b)$ of v- and e-bubbles such that the naturally ordered v-bubbles of $v_1$ precede the e-bubble of $v_1v_2$, which precedes the naturally ordered v-bubbles of $v_2$, and so on.
We call this sequence the \emph{splitting sequence} of $H$ and write $H=(H_1,\ldots,H_b)$. 
The splitting sequence of a single-vertex $1$-bubble is empty. 
Every $1$-bubble with more than one vertex is a v-bubble or a bouquet of several v-bubbles. 
The splitting sequence of a $2$-bubble may consist of several v- and e-bubbles. 
See Figure \ref{fig:bubble} for an illustration. 

\begin{figure}[t]
\begin{center}
\psfrag{v}{$v$}
\psfrag{v1}{$v^1$}
\psfrag{vk}{$v^k$}
\psfrag{vl}{$v^{\ell}$}
\psfrag{v_0}{$v_0$}
\psfrag{v01}{$v_0^1$}
\psfrag{v_1}{$v_1$}
\psfrag{v_n-1}{$v_{n-1}$}
\psfrag{vn-1k-1}{$v_{n-1}^{k_{n-1}}$}
\psfrag{v_n}{$v_n$}
\psfrag{u}{$u$}
\psfrag{u1}{$u^1$}
\includegraphics[scale=1.8]{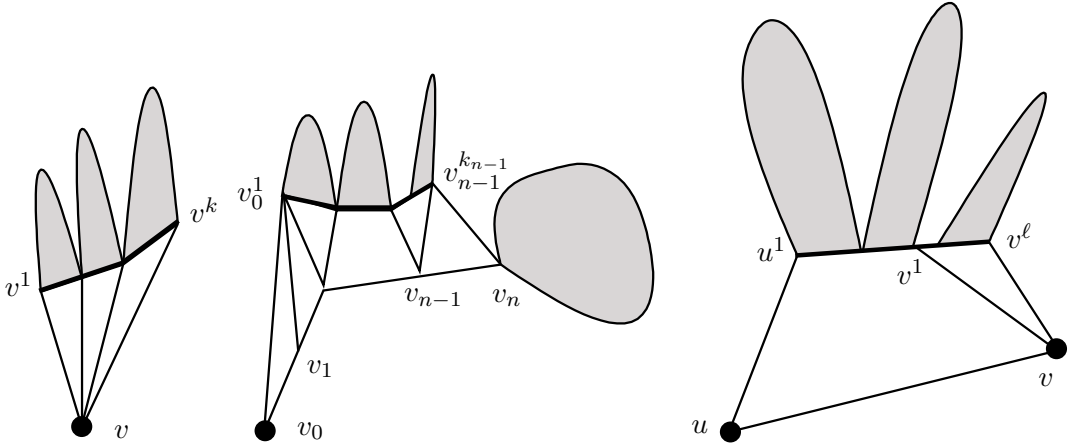}
\caption{Three ways of obtaining smaller bubbles from v- and e-bubbles as described in Lemma~\ref{lem:induction-structure}. 
The new bubbles are marked gray, and the new root-paths are drawn thick.}
\label{fig:nextbubble}
\end{center}
\end{figure}

The general structure of the induction in our proof is covered by the following lemma (see Figure \ref{fig:nextbubble}):

\begin{lemma}\label{lem:induction-structure}
\leavevmode
\begin{enumeratedlemma}
\item\label{item:v-bubble-without-a-root}
Let\/ $H$ be a v-bubble rooted at\/ $v$. 
Let\/ $v^1,\ldots,v^k$ be the neighbors of\/ $v$ in\/ $H$ from left to right. 
It follows that\/ $H-v$ is a\/ $k$-bubble with root-sequence\/ $v^1,\ldots,v^k$.
\item\label{item:v-bubble-without-a-path}
Let\/ $H$ be a v-bubble rooted at\/ $v_0$. 
Consider an induced path\/ $v_0\ldots v_n$ in\/ $H$ that starts with the rightmost edge at\/ $v_0$ and continues counterclockwise along the outer face of\/ $H$ so that\/ $v_1,\ldots,v_{n-1}$ are not cut-vertices in\/ $H$. 
It follows that\/ $H-\set{v_0,\ldots,v_n}$ has a unique component\/ $H'$ adjacent to both\/ $v_0$ and\/ $v_n$. 
Moreover, let\/ $X$ be the subgraph of\/ $H$ induced on\/ $v_0,\ldots,v_n$ and the vertices of\/ $H'$. 
Let\/ $v_0^1,\ldots,v_0^{k_0},v_1$ be the neighbors of\/ $v_0$ in\/ $X$ in clockwise order. 
For\/ $1\le i\le n-1$, let\/ $v_{i-1},v_i^1,\ldots,v_i^{k_i},v_{i+1}$ be the neighbors of\/ $v_i$ in\/ $X$ in clockwise order. 
Let\/ $v_{n-1},v_n^1,\ldots,v_n^{k_n}$ be the neighbors of\/ $v_n$ in\/ $X$ in clockwise order.
It follows that\/ $H'$ is a bubble with root-sequence\/ $v_0^1,\ldots,v_0^{k_0},\allowbreak v_1^1,\ldots,v_1^{k_1},\ldots,\allowbreak v_n^1,\ldots,v_n^{k_n}$ in which\/ $v_i^{k_i}$ and\/ $v_{i+1}^1$ coincide whenever the inner face of\/ $H$ containing\/ $v_iv_{i+1}$ is a triangle.
\item\label{item:e-bubble}
Let\/ $H$ be an e-bubble with roots\/ $u$ and\/ $v$. 
Let\/ $u^1,\ldots,u^k,v$ be the neighbors of\/ $u$ in\/ $H$ from left to right and\/ 
$u,v^1,\ldots,v^\ell$ be the neighbors of\/ $v$ in\/ $H$ from left to right. 
It follows that\/ $H-\set{u,v}$ is a bubble with root-sequence\/ $u^1,\ldots,u^k,\allowbreak v^1,\ldots,v^\ell$ in which\/ $u^k$ and\/ $v^1$ coincide if the inner face of\/ $H$ containing\/ $uv$ is a triangle.
\end{enumeratedlemma}
\end{lemma}

\begin{proof}
First we prove \ref{item:v-bubble-without-a-root}. 
Since $H$ is a v-bubble, $H-v$ is connected. 
The symmetric difference of the inner faces of $H$ incident to $v$, considered as sets of edges, gives a simple clockwise cycle in $H$ passing through $v$ and $v^1,\ldots,v^k$ in this order. 
Let $\gamma$ be a closed curve going counterclockwise from $v^k$ through the outer face of $G$ to $v^1$, and then through the inner faces of $H$ at $v$, crossing the edges $vv^2,\ldots,vv^{k-1}$ in this order, back to $v_k$. 
Clearly, $\gamma$ is a snip defining the bubble $H-v$ with root-sequence $v^1,\ldots,v^k$. 

Next we prove \ref{item:v-bubble-without-a-path}. 
Since none of $v_1,\ldots,v_{n-1}$ is a cut-vertex in $H$, the graph $H-\set{v_0,\ldots,v_n}$ has a component adjacent to both $v_0$ and $v_n$. 
Moreover, since the path $v_0\ldots v_n$ consists only of outer edges, such a component is unique. 
Thus $H'$ is well defined. 
Now, the symmetric difference of the inner faces of $X$ incident to any of $v_0,\ldots,v_n$, considered as sets of edges, gives a simple clockwise cycle in $H$ passing through $v_n,\ldots,v_0$ and then through $v_0^1,\ldots,v_0^{k_0},\allowbreak v_1^1,\ldots,v_1^{k_1},\ldots,\allowbreak v_n^1,\ldots,v_n^{k_n}$ in this order. 
Let $\gamma$ be a closed curve going counterclockwise from $v_n^{k_n}$ through the outer face of $G$ to $v_0^1$, and then through the inner faces of $H$ at $v_0,\ldots,v_n$, crossing the edges $v_0v_0^2,\ldots,v_0v_0^{k_0},\allowbreak v_1v_1^1,\ldots,v_1v_1^{k_1},\ldots,\allowbreak v_nv_n^1,\ldots,v_nv_n^{k_n-1}$ in this order, back to $v_n^{k_n}$. 
Clearly, $\gamma$ is a snip defining the bubble $H'$ with root-sequence $v_0^1,\ldots,v_0^{k_0},\allowbreak v_1^1,\ldots,v_1^{k_1},\ldots,\allowbreak v_n^1,\ldots,v_n^{k_n}$. 

Finally we show \ref{item:e-bubble}. 
Since $H$ is an e-bubble, $H-\set{u,v}$ is connected. 
Again, the symmetric difference of the inner faces of $H$ incident to $u$ or $v$, considered as sets of edges, gives a simple clockwise cycle in $H$ passing through $v$, $u$, and $v^1,\ldots,v^k$ in this order. 
Let $\gamma$ be a closed curve going counterclockwise from $v^\ell$ through the outer face of $G$ to $u^1$, and then through the inner faces of $H$ at $u$ and $v$, crossing the edges $uu^2,\ldots,uu^k,\allowbreak vv^1,\ldots,vv^{\ell-1}$ in this order, back to $v_\ell$. 
Clearly, $\gamma$ is a snip defining the bubble $H-\set{u,v}$ with root-sequence $u^1,\ldots,u^k,\allowbreak v^1,\ldots,v^\ell$.
\end{proof}

\section{Bounding regions}

Depending on the maximum degree $\Delta$ of $G$, define the set $S$ of $\Delta-1$ slopes to consist of the horizontal slope and the slopes of vectors $\ff_1,\ldots,\ff_{\Delta-2}$, where
\begin{gather*}
\ff_i=(-\tfrac{1}{2}+\tfrac{i-1}{\Delta-3},1)\quad\text{for $i=1,\ldots,\Delta-2$}.
\end{gather*}
An important property of $S$ is that it cuts the horizontal segment $L$ from $(-\frac{1}{2},1)$ to $(\frac{1}{2},1)$ into $\Delta-3$ segments 
of equal length $\frac{1}{\Delta-3}$.
We construct an outerplanar straight-line drawing of $G$ using only slopes from $S$ and preserving the given cyclic ordering of edges at each vertex of $G$.

The essential tool in proving that our construction does not make bubbles overlap are bounding regions. 
Their role is to bound the area of the plane occupied by bubbles. 
The bounding region of a bubble is parametrized by $\ell$ and $r$ which depend on the degrees of the roots in the bubble. 
Let $v$ be a point in the plane. 
For a vector $x$, let $\Ray{v}{x}=\set{v+\alpha x\colon\alpha\ge 0}$. 
For $0\le\ell\le\Delta-1$, we define $\LBR{v}{\ell}$ to be the cone consisting of $v$ and all points $p$ such that
\begin{itemizedcases}
\itemcase{$p_y\ge v_y$,}{}
\itemcase{$p$ lies on $\Ray{v}{\ff_1}$ or to the right of it}{if $\ell=1$,}
\itemcase{$p_x>v_x$}{if $\Delta=4$ and $\ell=2$,}
\itemcase{$p$ lies to the right of $\Ray{v}{\ff_\ell+\tfrac{1}{\Delta-4}\ff_1}$}{if $\Delta\ge 5$ and $2\le\ell\le\Delta-2$,}
\itemcase{$p$ lies to the right of $\Ray{v}{\ff_{\Delta-2}}$}{if $\ell=\Delta-1$.}
\end{itemizedcases}
Similarly, for $0\le r\le\Delta-1$, we define $\RBR{v}{r}$ to be the cone consisting of $v$ and all points $p$ such that
\begin{itemizedcases}
\itemcase{$p_y\ge v_y$,}{}
\itemcase{$p$ lies to the left of $\Ray{v}{\ff_1}$}{if $r=0$,}
\itemcase{$p_x<v_x$}{if $\Delta=4$ and $r=1$,}
\itemcase{$p$ lies to the left of $\Ray{v}{\ff_r+\tfrac{1}{\Delta-4}\ff_{\Delta-2}}$}{if $\Delta\ge 5$ and $1\le r\le\Delta-3$,}
\itemcase{$p$ lies on $\Ray{v}{\ff_{\Delta-2}}$ or to the left of it}{if $r=\Delta-2$.}
\end{itemizedcases}

\begin{figure}[t]
\begin{center}
\begin{tikzpicture}[>=latex,scale=2.5]
    \fill[color=black!5] (-0.85,1.7)--(0,0)--(-0.472,1.7)--cycle;
    \fill[color=black!10] (-0.472,1.7)--(0,0)--(-0.094,1.7)--cycle;
    \fill[color=black!15] (-0.094,1.7)--(0,0)--(0.283,1.7)--cycle;
    \fill[color=black!20] (0.283,1.7)--(0,0)--(0.85,1.7)--cycle;
    \fill[color=black!25] (0.85,1.7)--(0,0)--(2.8,0)--(1.95,1.7)--cycle;
    \draw (-0.85,0)--(3.65,0);
    \draw[thick,->] (0,0)--(-0.5,1);
    \draw[thick,->] (0,0)--(-0.167,1);
    \draw[thick,->] (0,0)--(0.167,1);
    \draw[thick,->] (0,0)--(0.5,1);
    \draw[->] (-0.5,1)--(-0.75,1.5);
    \draw[->] (-0.167,1)--(-0.417,1.5);
    \draw[->] (0.167,1)--(-0.083,1.5);
    \draw[->] (0.5,1)--(0.25,1.5);
    \draw (-0.75,1.5)--(-0.85,1.7);
    \draw[dashed] (0,0)--(-0.472,1.7);
    \draw[dashed] (0,0)--(-0.094,1.7);
    \draw[dashed] (0.167,1)--(0.283,1.7);
    \draw[dashed] (0.5,1)--(0.85,1.7);
    \node[above] at (-0.85,0) {\hbox to 0pt{\hss$\ell={}$}$0$};
    \node[above] at (-0.85,1.7) {\hbox to 0pt{\hss$\ell={}$}$1$};
    \node[above] at (-0.472,1.7) {$2$};
    \node[above] at (-0.094,1.7) {$3$};
    \node[above] at (0.283,1.7) {$4$};
    \node[above] at (0.85,1.7) {$5$};
    \node[above right] at (0,0) {\hskip 0.4em $u$};
    \begin{scope}[xshift=2.8cm,xscale=-1]
    \fill[color=black!5] (-0.85,1.7)--(0,0)--(-0.472,1.7)--cycle;
    \fill[color=black!10] (-0.472,1.7)--(0,0)--(-0.094,1.7)--cycle;
    \fill[color=black!15] (-0.094,1.7)--(0,0)--(0.283,1.7)--cycle;
    \fill[color=black!20] (0.283,1.7)--(0,0)--(0.85,1.7)--cycle;
    \draw[thick,->] (0,0)--(-0.5,1);
    \draw[thick,->] (0,0)--(-0.167,1);
    \draw[thick,->] (0,0)--(0.167,1);
    \draw[thick,->] (0,0)--(0.5,1);
    \draw[->] (-0.5,1)--(-0.75,1.5);
    \draw[->] (-0.167,1)--(-0.417,1.5);
    \draw[->] (0.167,1)--(-0.083,1.5);
    \draw[->] (0.5,1)--(0.25,1.5);
    \draw (-0.75,1.5)--(-0.85,1.7);
    \draw[dashed] (0,0)--(-0.472,1.7);
    \draw[dashed] (0,0)--(-0.094,1.7);
    \draw[dashed] (0.167,1)--(0.283,1.7);
    \draw[dashed] (0.5,1)--(0.85,1.7);
    \node[above] at (-0.85,0) {$r$\hbox to 0pt{${}=5$\hss}};
    \node[above] at (-0.85,1.7) {$4$};
    \node[above] at (-0.472,1.7) {$3$};
    \node[above] at (-0.094,1.7) {$2$};
    \node[above] at (0.283,1.7) {$1$};
    \node[above] at (0.85,1.7) {\hbox to 0pt{\hss$r={}$}$0$};
    \node[above left] at (0,0) {$v$\hskip 0.4em\null};
    \end{scope}
\end{tikzpicture}
\caption{Boundaries of $\LBR{u}{\ell}$ (left) and $\RBR{v}{r}$ (right) for $\Delta=6$. 
Vectors $\ff_i$ at $u$ and $v$ are indicated by thick arrows. 
Vectors $\tfrac{1}{2}\ff_1$ at $u+\ff_i$ and $\tfrac{1}{2}\ff_4$ at $v+\ff_i$ are indicated by thin arrows. 
Note that $u+\ff_3$ lies on the boundary of $\LBR{u}{4}$ and $v+\ff_2$ lies on the boundary of $\RBR{v}{1}$.}
\label{fig:br-delta>=5}
\end{center}
\end{figure}
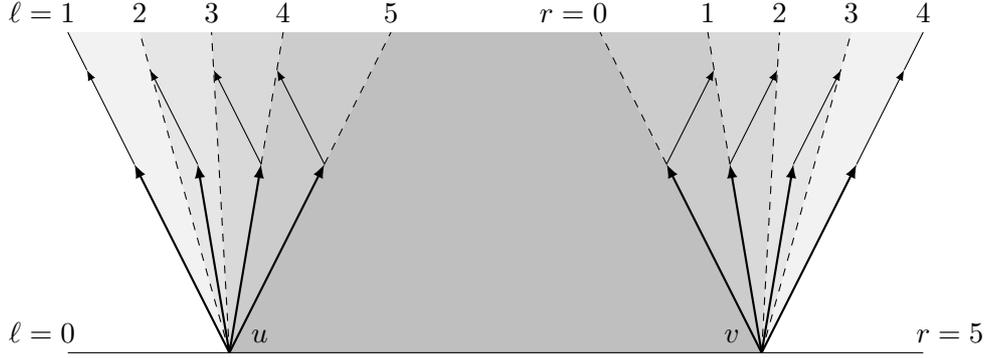

\noindent See Figure \ref{fig:br-delta>=5} for an illustration. 
Now, for points $u,v$ in the plane such that $u_y=v_y$ and $u_x\le v_x$, we define bounding regions as follows:
\begin{alignat*}{2}
\BR{uv}{\ell}{r}&=\LBR{u}{\ell}\cap\RBR{v}{r}&\quad&\text{for $0\le\ell,r\le\Delta-1$,}\\
\BRU{uv}{\ell}{r}{h}&=\BR{uv}{\ell}{r}\cap\set{p\colon p_y<u_y+h}&\quad&\text{for $0\le\ell,r\le\Delta-1$ and $h>0$.}
\end{alignat*}
We denote $\BR{vv}{\ell}{r}$ simply by $\BR{v}{\ell}{r}$ and $\BRU{vv}{\ell}{r}{h}$ simply by $\BRU{v}{\ell}{r}{h}$. 
Note that the bottom border of a bounding region is always included, the left border (if exists) is included if $\ell=1$, the right border (if exists) is included if $r=\Delta-2$, and the top border (if exists) is never included. 
If $\ell>r$, then $\BR{v}{\ell}{r}=\set{v}$.

We use $\BR{v}{\ell}{r}$ and $\BRU{v}{\ell}{r}{h}$ to bound drawings of $1$-bubbles $H$ with root $v$ such that $r-\ell+1=d_H(v)$. 
Note that every $1$-bubble drawn inside $\BR{v}{\ell}{r}$ can be scaled to fit inside $\BRU{v}{\ell}{r}{h}$ for any $h>0$ without changing slopes. 
We use $\BR{uv}{\ell}{r}$ and $\BRU{uv}{\ell}{r}{h}$ with $u\ne v$ to bound drawings of $2$-bubbles $H$ whose root-path starts at $u$ and ends at $v$, such that $\ell=\Delta-d_H(u)$ and $r=d_H(v)-1$. 
Here $H$ cannot be scaled if the positions of both $u$ and $v$ are fixed, so the precise value of $h$ matters. 
However, every $2$-bubble drawn inside $\BR{uv}{\ell}{r}$ can be scaled to fit inside $\BRU{uw}{\ell}{r}{h}$ for any $h>0$ without changing slopes, where $w$ is some point of the segment $uv$. 

\begin{lemma}
Bounding regions have the following geometric properties (\/$*$ stands for any value which if generally feasible is irrelevant to the statement).
\begin{enumeratedlemma}
\item\label{item:br-inclusion}
If\/ $u'v'\subset uv$, $\ell'\ge\ell$, and\/ $r'\le r$, then\/ $\BR{u'v'}{\ell'}{r'}\subset\BR{uv}{\ell}{r}$.
\item\label{item:br-slopes}
If\/ $i<\ell$, then a vector at\/ $u$ in direction\/ $\ff_i$ points outside\/ $\BR{uv}{\ell}{*}$ to the left of it.
If\/ $i>r$, then a vector at\/ $v$ in direction\/ $\ff_i$ points outside\/ $\BR{uv}{*}{r}$ to the right of it.
\item\label{item:br-neighbors}
If\/ $u,v,w$ are points on a horizontal line in this order from left to right and\/ $\ell-1\ge r+1$, then\/ $\BR{uv}{*}{r}\cap\BR{vw}{\ell}{*}=\set{v}$.
\enumeratedlemmaintertext{Moreover, the following holds for\/ $\Delta\ge 5$.}
\item\label{item:br-induction}
For\/ $\ell<r$, $h>0$, $u'=u+h\ff_\ell$, and\/ $v'=v+h\ff_r$, we have\/ $\BRU{u'v'}{1}{\Delta-2}{\tfrac{h}{\Delta-4}}\subset\BRU{uv}{\ell}{r}{\tfrac{\Delta-3}{\Delta-4}h}$.
\item\label{item:br-composition}
If\/ $u,v,w$ are points on a horizontal line in this order from left to right, $r'\le\Delta-3$, and\/ $r\ge 1$, then\/ $\BRU{uv}{\ell}{r'}{(\tfrac{\Delta-3}{\Delta-4})^2\abs{vw}}\subset\BR{uw}{\ell}{r}$.
\item\label{item:br-nonneighbors}
If\/ $u,v,w,x$ are points on a horizontal line in this order from left to right, $\abs{uv}=\abs{wx}\le\abs{vw}$, $\ell\ge 2$, and\/ $r\le\Delta-3$, then\/ $\BRU{uv}{*}{r}{\tfrac{\Delta-3}{\Delta-4}\abs{uv}}\cap\BRU{wx}{\ell}{*}{\tfrac{\Delta-3}{\Delta-4}\abs{wx}}=\emptyset$.
\end{enumeratedlemma}
\end{lemma}

\begin{proof}
Statement \ref{item:br-inclusion} clearly follows from the definition. 
Statement \ref{item:br-slopes} is implied by the definition and for $\Delta\ge 5$ by the fact that
\begin{alignat*}{3}
&\ff_{\ell}+\tfrac{1}{\Delta-4}\ff_1&&=\ff_{{\ell}-1}+\tfrac{1}{\Delta-4}\ff_{\Delta-3}&\quad&\text{for ${\ell}=2,\ldots,\Delta-2$},\\
&\ff_r+\tfrac{1}{\Delta-4}\ff_{\Delta-2}&&=\ff_{r+1}+\tfrac{1}{\Delta-4}\ff_2&\quad&\text{for $r=1,\ldots,\Delta-3$}.
\end{alignat*}
Statement \ref{item:br-slopes} directly yields \ref{item:br-neighbors}: the vector at $v$ in direction $\ff_{r+1}$ points outside both $\BR{uv}{*}{r}$ and $\BR{vw}{\ell}{*}$. 
To see \ref{item:br-induction}, note that the point $u'+\tfrac{h}{\Delta-4}\ff_1$, which is the top-left corner of $\BRU{u'v'}{1}{\Delta-2}{\tfrac{h}{\Delta-4}}$, equals $u+h(\ff_\ell+\tfrac{1}{\Delta-4}\ff_1)$. 
Hence, it lies at the top-left corner of $\BRU{uv}{\ell}{r}{\tfrac{\Delta-3}{\Delta-4}h}$. 
Similarly, the top-right corner of $\BRU{u'v'}{1}{\Delta-2}{\tfrac{h}{\Delta-4}}$ lies at the top-right corner of $\BRU{uv}{\ell}{r}{\tfrac{\Delta-3}{\Delta-4}h}$. 
To prove \ref{item:br-composition}, it suffices to consider the case $r=1$ and $r'=\Delta-3$. The top-right corner of $\BRU{uv}{\ell}{r'}{(\tfrac{\Delta-3}{\Delta-4})^2\abs{vw}}$ is
\begin{equation*}
\begin{split}
v+\tfrac{\Delta-3}{\Delta-4}\abs{vw}(\ff_{\Delta-3}+\tfrac{1}{\Delta-4}\ff_{\Delta-2})&=v+\abs{vw}(\ff_{\Delta-2}+\tfrac{1}{\Delta-4}\ff_1+\tfrac{\Delta-3}{\Delta-4}\cdot\tfrac{1}{\Delta-4}\ff_{\Delta-2})\\
&=w+\tfrac{\Delta-3}{\Delta-4}\abs{vw}(\ff_1+\tfrac{1}{\Delta-4}\ff_{\Delta-2}).
\end{split}
\end{equation*}
Therefore, it lies on the right side of $\BR{uw}{\ell}{1}$, and the conclusion of \ref{item:br-composition} follows. 
Finally, for the proof of \ref{item:br-nonneighbors}, it suffices to consider the case $\ell=2$, $r=\Delta-3$, and $\abs{uv}=\abs{vw}=\abs{wx}=\lambda$. 
The top-right corner of $\BRU{uv}{*}{\Delta-3}{\tfrac{\Delta-3}{\Delta-4}\lambda}$ and the top-left corner of $\BRU{wx}{2}{*}{\tfrac{\Delta-3}{\Delta-4}\lambda}$ are respectively
\begin{alignat*}{5}
v&+\lambda(\ff_{\Delta-3}&&+\tfrac{1}{\Delta-4}\ff_{\Delta-2})&&={}&v&+\lambda(\ff_{\Delta-2}&&+\tfrac{1}{\Delta-4}\ff_2),\\
w&+\lambda(\ff_2&&+\tfrac{1}{\Delta-4}\ff_1)&&={}&w&+\lambda(\ff_1&&+\tfrac{1}{\Delta-4}\ff_{\Delta-3}).
\end{alignat*}
They coincide if $\Delta=5$, otherwise the former lies to the left of the latter.
\end{proof}

\section{The drawing}

We present the construction of a drawing first for $\Delta\ge 5$ and then for $\Delta=4$. 
Both constructions follow the same idea but differ in technical details. 
The difference comes from the fact that any bubble can be drawn inside a bounding region of bounded height (independent of the size of the bubble) when $\Delta\ge 5$, but not when $\Delta=4$. 

\begin{lemma}\label{lem:Delta>=5}
Suppose\/ $\Delta\ge 5$.
\begin{enumeratedlemma}
\item\label{item:1-bubble}
Let\/ $H$ be a\/ $1$-bubble with root\/ $v$ such that\/ $d_H(v)\le\Delta-1$. 
Suppose that the position of\/ $v$ is fixed. 
Let\/ $\ell$ and\/ $r$ be such that\/ $0\le\ell,r\le\Delta-1$ and\/ $r-\ell+1=d_H(v)$. 
It follows that there is a straight-line drawing of\/ $H$ inside\/ $\BR{v}{\ell}{r}$.
\item\label{item:2-bubble}
Let\/ $H$ be a\/ $2$-bubble with first root\/ $u$ and second root\/ $v$ such that\/ $d_H(u),d_H(v)\le\Delta-1$. 
Suppose that the positions of\/ $u$ and\/ $v$ are fixed on a horizontal line in this order from left to right. 
Let\/ $\ell=\Delta-d_H(u)$ and\/ $r=d_H(v)-1$. 
It follows that there is a straight-line drawing of\/ $H$ inside\/ $\BRU{uv}{\ell}{r}{\tfrac{\Delta-3}{\Delta-4}\abs{uv}}$ such that 
the root-path of\/ $H$ is drawn as the segment\/ $uv$.
\item\label{item:k-bubble}
Let\/ $H$ be a\/ $k$-bubble with root-sequence\/ $v_1,\ldots,v_k$. 
If\/ $k=1$, then suppose\/ $d_H(v_1)\le\Delta-2$, otherwise suppose\/ $d_H(v_1),d_H(v_k)\le\Delta-1$. 
Suppose that the positions of\/ $v_1,\ldots,v_k$ are fixed in this order from left to right on a horizontal line so that\/ $\abs{v_1v_2}=\ldots=\abs{v_{k-1}v_k}=\lambda$, for some\/ $\lambda>0$. 
It follows that there is a straight-line drawing of\/ $H$ inside\/ $\BRU{v_1v_k}{1}{\Delta-2}{\tfrac{\Delta-3}{\Delta-4}\lambda}$ such that the root-path of\/ $H$ is drawn as the segment\/ $v_1v_k$.
\end{enumeratedlemma}
The drawings claimed above use only slopes from\/ $S$ and preserve the order of edges around each vertex\/ $w$ of\/ $H$ under the assumption that all edges connecting\/ $w$ to\/ $G-H$ (if exist) are drawn in the correct order outside the considered bounding region.
\end{lemma}

\begin{proof}
The proof constructs the required drawing by induction on the size of $H$. 
That is, to prove any of \ref{item:1-bubble}--\ref{item:k-bubble} for a bubble $H$, we assume that the entire lemma holds for any bubble with fewer vertices than $H$ has. 
The construction we are going to describe clearly preserves the order of edges at every vertex of $H$ and uses only slopes from $S$, and we do not explicitly state this observation anywhere further in the proof. 

\begin{enumeratedlemmaproofitem}{item:1-bubble}
We consider several cases depending on the values of $\ell$ and $r$ and on whether $H$ is a single v-bubble or a bouquet of several v-bubbles.

\begin{case}
\textit{$\ell>r$.}

In this case $v$ is the only vertex of $H$ and the statement is trivial. 
\end{case}

\begin{case}\label{case:v-bubble-simple}
\textit{$H$ is a v-bubble and\/ $1\le\ell\le r\le\Delta-2$.}

Define $H'=H-v$, and let $v^\ell,\ldots,v^r$ be the neighbors of $v$ in $H$ from left to right. 
By \ref{item:v-bubble-without-a-root}, the graph $H'$ is an $(r-\ell+1)$-bubble with root-sequence $v^\ell,\ldots,v^r$. 
Put each vertex $v^i$ at point $v+\ff_i$. 
Consider two subcases. 

\begin{subcase}\label{subcase:v-bubble-simple-l=r}
\textit{$1\le\ell=r\le\Delta-2$.}

By the induction hypothesis \ref{item:1-bubble}, the $1$-bubble $H'$ can be drawn inside $\BR{v^r}{0}{d_{H'}(v^r)-1}$ as well as inside $\BR{v^r}{\Delta-d_{H'}(v^r)}{\Delta-1}$. 
Choose the former drawing if $\ell=r=\Delta-2$, the latter if $\ell=r=1$, or any of the two otherwise. 
After appropriate scaling, the chosen drawing fits within $\BR{v}{r}{r}$. 
\end{subcase}

\begin{subcase}\label{subcase:v-bubble-simple-l<r}
\textit{$1\le\ell<r\le\Delta-2$.}

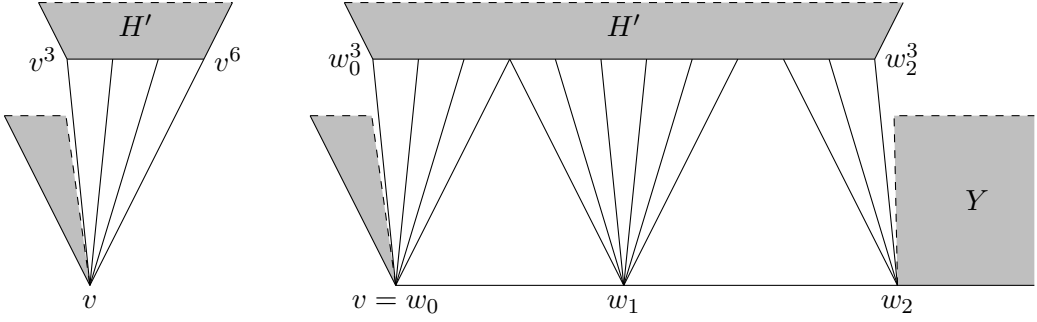
\begin{figure}[t]
\begin{center}
\begin{tikzpicture}[baseline,scale=3]
    \begin{scope}[scale=0.6]
    \fill[color=black!25] (0,0)--(-0.625,1.25)--(-0.175,1.25)--cycle;
    \draw (0,0)--(-0.625,1.25);
    \draw[dashed] (-0.625,1.25)--(-0.175,1.25)--(0,0);
    \end{scope}
    \fill[color=black!25] (-0.1,1)--(-0.225,1.25)--(0.625,1.25)--(0.5,1)--cycle;
    \draw (0,0)--(-0.1,1)--(-0.225,1.25);
    \draw (0,0)--(0.1,1);
    \draw (0,0)--(0.3,1);
    \draw (0,0)--(0.625,1.25);
    \draw (-0.1,1)--(0.5,1);
    \draw[dashed] (-0.225,1.25)--(0.625,1.25);
    \node[below] at (0,0) {$v$};
    \node[left] at (-0.1,1) {$v^3$};
    \node[right] at (0.5,1) {$v^6$};
    \node at (0.2,1.15) {$H'$};
\end{tikzpicture}
\hspace{.5cm}
\begin{tikzpicture}[baseline,scale=3]
    \begin{scope}[scale=0.6]
    \fill[color=black!25] (0,0)--(-0.625,1.25)--(-0.175,1.25)--cycle;
    \draw (0,0)--(-0.625,1.25);
    \draw[dashed] (-0.625,1.25)--(-0.175,1.25)--(0,0);
    \end{scope}
    \fill[color=black!25] (-0.1,1)--(-0.225,1.25)--(2.225,1.25)--(2.1,1)--cycle;
    \begin{scope}[xshift=2.2cm,scale=0.6]
    \fill[color=black!25] (0,0)--(-0.025,1.25)--(1,1.25)--(1,0)--cycle;
    \draw[dashed] (0,0)--(-0.025,1.25)--(1,1.25);
    \end{scope}
    \draw (0,0)--(2.8,0);
    \draw (0,0)--(-0.1,1)--(-0.225,1.25);
    \draw (0,0)--(0.1,1);
    \draw (0,0)--(0.3,1);
    \draw (0,0)--(0.5,1);
    \draw (1,0)--(0.5,1);
    \draw (1,0)--(0.7,1);
    \draw (1,0)--(0.9,1);
    \draw (1,0)--(1.1,1);
    \draw (1,0)--(1.3,1);
    \draw (1,0)--(1.5,1);
    \draw (2.2,0)--(1.7,1);
    \draw (2.2,0)--(1.9,1);
    \draw (2.2,0)--(2.1,1)--(2.225,1.25);
    \draw (-0.1,1)--(2.1,1);
    \draw[dashed] (-0.225,1.25)--(2.225,1.25);
    \node[below] at (0,0) {$v=w_0$};
    \node[below] at (1,0) {$w_1$};
    \node[below] at (2.2,0) {$w_2$};    
    \node[left] at (-0.1,1) {$w_0^3$};
    \node[right] at (2.1,1) {$w_2^3$};
    \node at (1,1.15) {$H'$};
    \node at (2.55,0.375) {$Y$};
\end{tikzpicture}
\caption{Sample drawings of $1$-bubbles for $\Delta=8$, illustrating Subcase \ref{subcase:v-bubble-simple-l<r} (left) and Subcase \ref{subcase:v-bubble-complex-l<r} (right) of the proof of \ref{item:1-bubble}. 
The gray areas labeled $H'$ and $Y$ denote bounding regions of the bubbles $H'$ and $Y$. 
The unlabeled gray areas are where other parts of a $1$-bubble may be added by repeating the procedure in Case \ref{case:1-bubble-ind} of the proof of \ref{item:1-bubble}.}
\label{fig:1-bubble}
\end{center}
\end{figure}

It follows that $v^\ell,\ldots,v^r$ lie on a common horizontal line $L$ and partition $L$ into segments of length $\tfrac{1}{\Delta-3}$. 
Apply the induction hypothesis \ref{item:k-bubble} to draw $H'$ inside $\BRU{v^\ell v^r}{1}{\Delta-2}{\tfrac{1}{\Delta-4}}$ (see Figure \ref{fig:1-bubble}, left). 
It follows from \ref{item:br-induction} that this bounding region is contained in $\BR{v}{\ell}{r}$. 
\end{subcase}
\end{case}

\begin{case}\label{case:v-bubble-complex}
\textit{$H$ is a v-bubble, $0\le\ell\le r\le\Delta-1$, and\/ $\ell=0$ or\/ $r=\Delta-1$.}

As $d_H(v)\le\Delta-1$, the cases $\ell=0$ and $r=\Delta-1$ cannot hold simultaneously. 
Therefore, by symmetry, it is enough to consider only the case that $1\le\ell\le r=\Delta-1$. 
Consider two subcases of the latter. 

\begin{subcase}\label{subcase:v-bubble-complex-l=r}
\textit{$\ell=r=\Delta-1$.}

It follows that $v$ has only one neighbor in $H$, say $w$, and $H'=H-v$ is a $1$-bubble rooted at $w$. 
Put $w$ horizontally to the right of $v$. 
Draw $H'$ inside $\BR{w}{\Delta-d_{H'}(w)}{\Delta-1}$ by the induction hypothesis \ref{item:1-bubble}, scaling the drawing appropriately to fit it within $\BR{v}{\Delta-1}{\Delta-1}$. 
\end{subcase}

\begin{subcase}\label{subcase:v-bubble-complex-l<r}
\textit{$1\le\ell<r=\Delta-1$.}

It follows that $v$ has at least two neighbors in $H$. 
Let $P=w_0\ldots w_n$ be the simple path of length $n\ge 1$ that stars at $w_0=v$ with the rightmost edge and continues counterclockwise along the outer face of $H$ so that
\begin{itemize}
\item the vertices $w_1,\ldots,w_{n-1}$ have degree $\Delta$ and are not cut-vertices in $H$,
\item the vertex $w_n$ has degree at most $\Delta-1$ or is a cut-vertex in $H$.
\end{itemize}
Note that the first condition is satisfied vacuously if $n=1$. 
Since the degrees of $w_1,\ldots,w_{n-1}$ are at least $3$ and by outerplanarity, $P$ is an induced path. 
Therefore, by \ref{item:v-bubble-without-a-path}, the graph $H-P$ has exactly one component $H'$ adjacent to both $w_0$ and $w_n$. 
All other components of $H-P$ are adjacent to $w_n$. 
Together with $w_n$ they form a (possibly trivial) $1$-bubble $Y$ rooted at $w_n$. 
Let $X$ denote the subgraph of $H$ induced on $w_0,\ldots,w_n$ and the vertices of $H'$. 
Define $r_X=d_X(w_n)-1$ and $\ell_Y=\Delta-d_Y(w_n)$. 
Let $w_0^{\ell},\ldots,w_0^{\Delta-2},w_1$ be the neighbors of $w_0$ in $X$ ordered clockwise. 
Let $w_{i-1},w_i^1,\ldots,w_i^{\Delta-2},w_{i+1}$ be the neighbors of $w_i$ in $X$ ordered clockwise, for $1\le i\le n-1$. 
Let $w_{n-1},w_n^1,\ldots,w_n^{r_X}$ be the neighbors of $w_n$ in $X$ ordered clockwise. 
It follows from \ref{item:v-bubble-without-a-path} that $H'$ is a bubble with root-sequence $w_0^{\ell},\ldots,w_0^{\Delta-2},\allowbreak w_1^1,\ldots,w_1^{\Delta-2},\ldots,\allowbreak w_n^1,\ldots,w_n^{r_X}$ in which $w_i^{\Delta-2}$ and $w_{i+1}^1$ coincide whenever the inner face of $H$ containing $w_iw_{i+1}$ is a triangle. 
For $i=0,\ldots,n-1$, define
\begin{equation*}
\lambda_i=\begin{cases}
1&\text{if $w_i^{\Delta-2}=w_{i+1}^1$},\\
\tfrac{\Delta-2}{\Delta-3}&\text{if $w_i^{\Delta-2}\ne w_{i+1}^1$}.
\end{cases}
\end{equation*}
Put the vertices $w_1,\ldots,w_n$ in this order from left to right on the horizontal line going through $w_0$ in such a way that $\abs{w_iw_{i+1}}=\lambda_i$ for $0\le i\le n-1$. 
Put each vertex $w_i^j$ at point $w_i+\ff_j$. 
Note that if $w_i^{\Delta-2}$ and $w_{i+1}^1$ are the same vertex, then they correctly end up at the same point. 
All $w_i^j$ lie on a common horizontal line $L$ at distance $1$ above the segment $w_0w_n$ and partition $L$ into segments of length $\tfrac{1}{\Delta-3}$. 
Define
\begin{gather*}
B_X=\BRU{w_0w_n}{\ell}{r_X}{\tfrac{\Delta-3}{\Delta-4}},\\
B_Y=\BRU{w_n}{\ell_Y}{\Delta-1}{1}.
\end{gather*}
Draw $H'$ inside $\BRU{w_0^{\ell}w_n^{r_X}}{1}{\Delta-2}{\tfrac{1}{\Delta-4}}$ using the induction hypothesis \ref{item:k-bubble} (see Figure \ref{fig:1-bubble}, right). 
Note that if $H'$ is a $1$-bubble, then $n=1$ and the root of $H'$ has at least two edges outside $H'$ (the ones going to $w_0$ and $w_1$), so \ref{item:k-bubble} can indeed be applied. 
By \ref{item:br-induction}, this bounding region is contained in $B_X$. 
Draw $Y$ inside $B_Y$ using the induction hypothesis and scaling. 
This way it lies entirely below the line $L$. 
By \ref{item:br-slopes}, the drawing of $Y$ lies to the right of the edge $w_nw_n^{r_X}$ and thus does not overlap with the drawing of $X$. 
Clearly, $B_X$ and $B_Y$ are contained in $\BR{w_0}{\ell}{\Delta-1}$. 
\end{subcase}
\end{case}

\begin{case}\label{case:1-bubble-ind}
\textit{$H$ consists of at least two v-bubbles.}

Let $(H_1,\ldots,H_b)$ be the splitting sequence of $H$. 
Thus all $H_1,\ldots,H_b$ are v-bubbles and $b\ge 2$. 
Define $X=(H_2,\ldots,H_b)$, $r_1=\ell+d_{H_1}(v)-1$, and $\ell_X=\Delta-d_X(v)$. 
By the induction hypothesis \ref{item:1-bubble}, we can draw $H_1$ inside $\BR{v}{\ell}{r_1}$ and $X$ inside $\BR{v}{\ell_X}{r}$. 
We scale the drawing of $H_1$ to make it so small that it lies entirely below the horizontal lines determined by all the vertices of $X$ other than $v$ and those lying on the horizontal line passing through $v$ (see Figure \ref{fig:1-bubble}). 
Since $r_1+1=\ell_X$ and by \ref{item:br-slopes}, our scaled drawing of $H_1$ lies to the left of the leftmost edge at the root $v$ of $X$. 
Thus the drawings of $H_1$ and $X$ do not overlap. 
By \ref{item:br-inclusion}, they both fit within $\BR{uv}{\ell}{r}$. 
\end{case}
\end{enumeratedlemmaproofitem}

\begin{enumeratedlemmaproofitem}{item:2-bubble}
We consider several cases depending on the bubbles forming the splitting sequence of $H$. 
The cases are not pairwise disjoint, but they cover all possible situations. 

\begin{case}\label{case:bridge}
\textit{The splitting sequence of\/ $H$ contains a trivial single-edge e-bubble.}

Let $(H_1,\ldots,H_b)$ be the splitting sequence of $H$, and let $H_i$ be a trivial e-bubble. 
Let $u_i$ and $v_i$ be the first and the second roots of $H_i$, respectively. 
That is, $H_i$ consists only of the edge $u_iv_i$.
Define $X=(H_1,\ldots,H_{i-1})$ and $Y=(H_{i+1},\ldots,H_b)$.
If $i=1$, then $X$ is a trivial one-vertex bubble, while if $i=b$, then $Y$ is a trivial one-vertex bubble. 
Choose a small $\lambda>0$. 
If $X$ is a $1$-bubble, then $u_i=u$, and otherwise put $u_i$ on the segment $uv$ so that $\abs{uu_i}=\lambda$. 
Similarly, if $Y$ is a $1$-bubble, then $v_i=v$, and otherwise put $v_i$ on $uv$ so that $\abs{v_iv}=\lambda$. 
Draw $X$ inside $\BRU{uu_i}{\ell}{d_X(u_i)-1}{\frac{\Delta-3}{\Delta-4}\lambda}$ and $Y$ inside $\BRU{v_iv}{\Delta-d_Y(v_i)}{r}{\frac{\Delta-3}{\Delta-4}\lambda}$ using the induction hypothesis \ref{item:1-bubble} or \ref{item:2-bubble}. 
Clearly, both bounding regions are contained in $\BRU{uv}{\ell}{r}{\tfrac{\Delta-3}{\Delta-4}\abs{uv}}$. 
Moreover, if $\lambda$ has been chosen small enough, then the two bounding regions are disjoint. 
\end{case}

\begin{case}\label{case:e-bubble}
\textit{$H$ is a non-trivial e-bubble.}

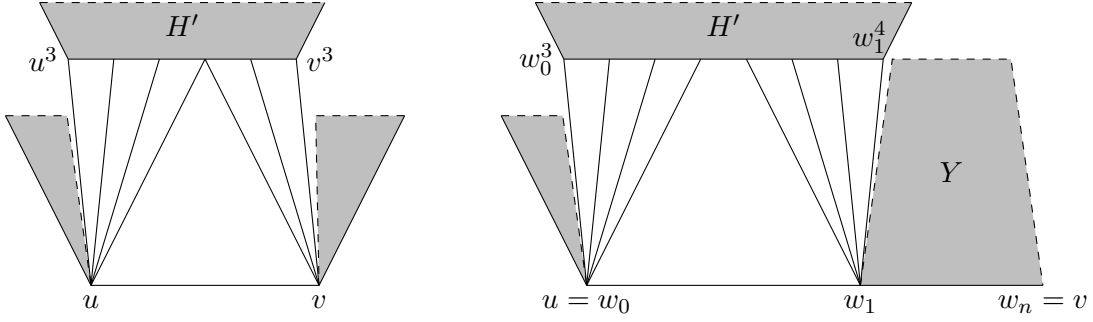
\begin{figure}[t]
\begin{center}
\begin{tikzpicture}[baseline,scale=3]
    \begin{scope}[scale=0.6]
    \fill[color=black!25] (0,0)--(-0.625,1.25)--(-0.175,1.25)--cycle;
    \draw (0,0)--(-0.625,1.25);
    \draw[dashed] (-0.625,1.25)--(-0.175,1.25)--(0,0);
    \end{scope}
    \begin{scope}[xshift=1cm,scale=0.6]
    \fill[color=black!25] (0,0)--(0.625,1.25)--(-0.025,1.25)--cycle;
    \draw (0,0)--(0.625,1.25);
    \draw[dashed] (0.625,1.25)--(-0.025,1.25)--(0,0);
    \end{scope}
    \fill[color=black!25] (-0.1,1)--(-0.225,1.25)--(1.025,1.25)--(0.9,1)--cycle;
    \draw (0,0)--(1,0);
    \draw (0,0)--(-0.1,1)--(-0.225,1.25);
    \draw (0,0)--(0.1,1);
    \draw (0,0)--(0.3,1);
    \draw (0,0)--(0.5,1);
    \draw (1,0)--(0.5,1);
    \draw (1,0)--(0.7,1);
    \draw (1,0)--(0.9,1)--(1.025,1.25);
    \draw (-0.1,1)--(0.9,1);
    \draw[dashed] (-0.225,1.25)--(1.025,1.25);
    \node[below] at (0,0) {$u$};
    \node[below] at (1,0) {$v$};
    \node[left] at (-0.1,1) {$u^3$};
    \node[right] at (0.9,1) {$v^3$};
    \node at (0.4,1.15) {$H'$};
\end{tikzpicture}
\hspace{1cm}
\begin{tikzpicture}[baseline,scale=3]
    \begin{scope}[scale=0.6]
    \fill[color=black!25] (0,0)--(-0.625,1.25)--(-0.175,1.25)--cycle;
    \draw (0,0)--(-0.625,1.25);
    \draw[dashed] (-0.625,1.25)--(-0.175,1.25)--(0,0);
    \end{scope}
    \begin{scope}[xshift=1.2cm,scale=0.8]
    \fill[color=black!25] (0,0)--(0.175,1.25)--(0.825,1.25)--(1,0)--cycle;
    \draw[dashed] (0,0)--(0.175,1.25)--(0.825,1.25)--(1,0);
    \end{scope}
    \fill[color=black!25] (-0.1,1)--(-0.225,1.25)--(1.425,1.25)--(1.3,1)--cycle;
    \draw (0,0)--(2,0);
    \draw (0,0)--(-0.1,1)--(-0.225,1.25);
    \draw (0,0)--(0.1,1);
    \draw (0,0)--(0.3,1);
    \draw (0,0)--(0.5,1);
    \draw (1.2,0)--(0.7,1);
    \draw (1.2,0)--(0.9,1);
    \draw (1.2,0)--(1.1,1);
    \draw (1.2,0)--(1.3,1)--(1.425,1.25);
    \draw (-0.1,1)--(1.3,1);
    \draw[dashed] (-0.225,1.25)--(1.425,1.25);
    \node[below] at (0,0) {$u=w_0$};
    \node[below] at (1.2,0) {$w_1$};
    \node[below] at (2,0) {$w_n=v$};
    \node[left] at (-0.1,1) {$w_0^3$};
    \node[above,inner sep=2pt] at (1.3,1) {$w_1^4$\kern 10pt};
    \node at (0.6,1.15) {$H'$};
    \node at (1.6,0.5) {$Y$};
\end{tikzpicture}
\caption{Sample drawings of $2$-bubbles for $\Delta=8$, illustrating Case \ref{case:e-bubble} (left, with $u^6=v^1$) and Case \ref{case:composition} (right, with $w_0^6\ne w_1^1$) of the proof of \ref{item:2-bubble}. 
The gray areas labeled $H'$ and $Y$ denote bounding regions of the bubbles $H'$ and $Y$. 
The unlabeled gray areas are where initial and final v-bubbles may be added by repeating the procedure in Case \ref{case:2-bubble-ind} of the proof of \ref{item:2-bubble}.}
\label{fig:2-bubble}
\end{center}
\end{figure}

It follows that $d_H(u),d_H(v)\ge 2$ and therefore $\ell\le\Delta-2$ and $r\ge 1$. 
Let $u^\ell,\ldots,u^{\Delta-2},v$ be the neighbors of $u$ in $H$ ordered clockwise, and let $u,v^1,\ldots,v^r$ be the neighbors of $v$ in $H$ ordered clockwise. 
Let $H'=H-\set{u,v}$. 
By \ref{item:e-bubble}, the graph $H'$ is a bubble with root-sequence $u^\ell,\ldots,u^{\Delta-2},\allowbreak v^1,\ldots,v^r$ in which $u^{\Delta-2}$ and $v^1$ coincide if the inner face of $H$ containing $uv$ is a triangle. 
Define 
\begin{equation*}
h=\begin{cases}
\abs{uv}&\text{if $u^{\Delta-2}=v^1$},\\
\frac{\Delta-3}{\Delta-2}\abs{uv}&\text{if $u^{\Delta-2}\ne v^1$}.
\end{cases}
\end{equation*}
Put each vertex $u^i$ at point $u+h\ff_i$ and each vertex $v^i$ at point $v+h\ff_i$. 
Note that if $u^{\Delta-2}$ and $v^1$ are the same vertex, then they correctly end up at the same point. 
All $u^i$ and $v^i$ lie on a common horizontal line and partition it into segments of length $\tfrac{h}{\Delta-3}$. 
Draw $H'$ inside $\BRU{u^\ell v^r}{1}{\Delta-2}{\tfrac{h}{\Delta-4}}$ using the induction hypothesis \ref{item:k-bubble} (see Figure \ref{fig:2-bubble}, left). 
This bounding region is contained in $\BRU{uv}{\ell}{r}{\tfrac{\Delta-3}{\Delta-4}h}$ by \ref{item:br-induction}. 
Since $h\le\abs{uv}$, we have $\BRU{uv}{\ell}{r}{\tfrac{\Delta-3}{\Delta-4}h}\subset\BRU{uv}{\ell}{r}{\tfrac{\Delta-3}{\Delta-4}\abs{uv}}$.
\end{case}

\begin{case}\label{case:composition}
\textit{The splitting sequence of\/ $H$ starts with a non-trivial e-bubble and contains some other e-bubbles but no trivial e-bubbles.}

Let $(H_1,\ldots,H_b)$ be the splitting sequence of $H$ and $w_0\ldots w_n$ be the root-path of $H$. 
Thus $b\ge 2$, $w_0=u$, $w_n=v$, and $n\ge 2$. 
We split $H$ into the e-bubble $H_1$ with roots $w_0$ and $w_1$ and the rest $Y=(H_2,\ldots,H_b)$ being a $2$-bubble with roots $w_1$ and $w_n$. 
Define $r_1=d_{H_1}(w_1)-1$ and $\ell_Y=\Delta-d_Y(w_1)$. 
Since $H_1,\ldots,H_b$ are non-trivial, we have $d_{H_1}(w_0),d_{H_1}(w_1),d_Y(w_1),d_Y(w_n)\ge 2$ and therefore $\ell\le\Delta-2$, $1\le r_1\le\Delta-3$, $2\le\ell_Y\le\Delta-2$, and $r\ge 1$. 
Let $w_0^\ell,\ldots,w_0^{\Delta-2},w_1$ be the neighbors of $w_0$ in $H_1$ ordered clockwise. 
Similarly, let $w_0,w_1^1,\ldots,w_1^{r_1}$ be the neighbors of $w_1$ in $H_1$ ordered clockwise. 
Note that $w_0^{\Delta-2}$ and $w_1^1$ are the same vertex if the inner face of $H$ containing $w_0w_1$ is a triangle. 
Define
\begin{equation*}
\alpha=\begin{cases}
1&\text{if $w_0^{\Delta-2}=w_1^1$,}\\
\tfrac{\Delta-3}{\Delta-2}&\text{if $w_0^{\Delta-2}\ne w_1^1$.}
\end{cases}
\end{equation*}
Fix the position of $w_1$ on the segment $w_0w_n$ so that $\alpha\abs{w_0w_1}=\tfrac{\Delta-3}{\Delta-4}\abs{w_1w_n}=h$. 
Put each vertex $w_0^i$ at point $w_0+h\ff_i$ and each vertex $w_1^i$ at point $w_1+h\ff_i$. 
Note that if $w_0^{\Delta-2}$ and $w_1^1$ are the same vertex, then they correctly end up at the same point. 
All $w_0^i$ and $w_1^i$ lie on a common horizontal line $L$ at distance $h$ to the segment $w_0w_n$ and partition $L$ into segments of length $\tfrac{h}{\Delta-3}$. 
Define
\begin{gather*}
B_1=\BRU{w_0w_1}{\ell}{r_1}{\tfrac{\Delta-3}{\Delta-4}h},\\
B_Y=\BRU{w_1w_n}{\ell_Y}{r}{h}.
\end{gather*}
Let $H'=H_1-\set{w_0,w_1}$. 
By \ref{item:e-bubble}, the graph $H'$ is a bubble with root-sequence $w_0^\ell,\ldots,w_0^{\Delta-2},\allowbreak w_1^1,\ldots,w_1^{r_1}$ in which $w_0^{\Delta-2}$ and $w_1^1$ may coincide. 
Draw $H'$ inside $\BRU{w_0^\ell w_1^{r_1}}{1}{\Delta-2}{\tfrac{h}{\Delta-4}}$ using the induction hypothesis \ref{item:2-bubble} (see Figure \ref{fig:2-bubble}, right). 
By \ref{item:br-induction}, this bounding region is contained in $B_1$. 
Since $r_1\le\Delta-3$ and by \ref{item:br-composition}, we have $B_1\subset\BR{w_0w_n}{\ell}{r}$. 
This and the fact that $h\le\abs{w_0w_n}$ imply $B_1\subset\BRU{w_0w_n}{\ell}{r}{\tfrac{\Delta-3}{\Delta-4}\abs{w_0w_n}}$. 
To complete the drawing of $H$, apply the induction hypothesis \ref{item:2-bubble} to draw $Y$ inside $B_Y$. 
This way it lies entirely below $L$ and therefore does not overlap with the drawing of $H'$. 
By \ref{item:br-slopes}, it also lies to the right of the edge $w_1w_1^{r_1}$. 
Clearly, $B_Y\subset\BRU{w_0w_n}{\ell}{r}{\tfrac{\Delta-3}{\Delta-4}\abs{w_0w_n}}$. 
\end{case}

\begin{case}\label{case:2-bubble-ind}
\textit{The splitting sequence of\/ $H$ starts or ends with a v-bubble.}

The two cases are symmetric, so it is enough to consider only the case that the splitting sequence of $H$ starts with a v-bubble. 
Hence, let $(H_1,\ldots,H_b)$ be the splitting sequence of $H$, where $H_1$ is a v-bubble. 
Define $X=(H_2,\ldots,H_b)$, $r_1=\ell+d_{H_1}(u)-1$, and $\ell_X=\Delta-d_Y(u)$.
By the induction hypothesis \ref{item:1-bubble}, we can draw $H_1$ inside $\BR{u}{\ell}{r_1}$, and by the induction hypothesis \ref{item:2-bubble}, we can draw $X$ inside $\BRU{uv}{\ell_X}{r}{\tfrac{\Delta-3}{\Delta-4}\abs{uv}}$. 
We scale the drawing of $H_1$ to make it so small that it lies entirely below the horizontal lines determined by all the vertices of $X$ not lying on the root-path as well as below the horizontal line bounding from above the requested bounding region of $H$ (see Figure \ref{fig:2-bubble}). 
Since $r_1+1=\ell_X$ and by \ref{item:br-slopes}, our scaled drawing of $H_1$ lies to the left of the leftmost edge at the root $u$ of $X$. 
Thus the drawings of $H_1$ and $X$ do not overlap. 
By \ref{item:br-inclusion}, they both fit within $\BRU{uv}{\ell}{r}{\tfrac{\Delta-3}{\Delta-4}\abs{uv}}$. 
\end{case}
\end{enumeratedlemmaproofitem}

\begin{enumeratedlemmaproofitem}{item:k-bubble}
If $k=1$, then the claim follows directly from \ref{item:1-bubble} and \ref{item:br-inclusion} by scaling. 
Thus assume $k\ge 2$. 
There is a splitting of $H$ into $2$-bubbles $X_1,\ldots,X_{k-1}$ so that the splitting sequences of $X_1,\ldots,X_{k-1}$ together form the splitting sequence of $H$. 
In particular,
\begin{itemizedcases}
\itemcase{the roots of $X_i$ are $v_i$ and $v_{i+1}$}{for $i=1,\ldots,k-1$,}
\itemcase{$X_{i-1}\cap X_i=\set{v_i}$}{for $i=2,\ldots,k-1$.}
\end{itemizedcases}
Apply \ref{item:2-bubble} to draw each $X_i$ inside $\BRU{v_iv_{i+1}}{\Delta-d_{X_i}(v_i)}{d_{X_i}(v_{i+1})-1}{\tfrac{\Delta-3}{\Delta-4}\lambda}$. 
Consecutive bounding regions do not overlap by \ref{item:br-neighbors}, while non-consecutive ones are disjoint by \ref{item:br-nonneighbors}. 
By \ref{item:br-inclusion}, they are all contained in $\BRU{v_1v_k}{1}{\Delta-2}{\tfrac{\Delta-3}{\Delta-4}\lambda}$.\qedhere
\end{enumeratedlemmaproofitem}
\end{proof}

\begin{lemma}
Suppose\/ $\Delta=4$.
\begin{enumeratedlemma}
\item\label{item:1-bubble-Delta=4}
Let\/ $H$ be a\/ $1$-bubble with root\/ $v$ such that\/ $d_H(v)\le 3$. 
Suppose that the position of\/ $v$ is fixed. 
Let\/ $\ell$ and\/ $r$ be such that\/ $0\le\ell,r\le 3$ and\/ $r-\ell+1=d_H(v)$. 
It follows that there is a straight-line drawing of\/ $H$ inside\/ $\BR{v}{\ell}{r}$.
\item\label{item:2-bubble-Delta=4}
Let\/ $H$ be a\/ $2$-bubble with first root\/ $u$ and second root\/ $v$ such that\/ $d_H(u),d_H(v)\le 3$. 
Suppose that the positions of\/ $u$ and\/ $v$ are fixed on a horizontal line in this order from left to right. 
Let\/ $\ell=4-d_H(u)$ and\/ $r=d_H(v)-1$. 
It follows that there is a straight-line drawing of\/ $H$ inside\/ $\BR{uv}{\ell}{r}$ such that the root-path of\/ $H$ is drawn as the segment\/ $uv$.
\end{enumeratedlemma}
The drawings claimed above use only slopes from\/ $S$ and preserve the order of edges around each vertex\/ $w$ of\/ $H$ under the assumption that all edges connecting\/ $w$ to\/ $G-H$ (if exist) are drawn in the correct order outside the considered bounding region.
\end{lemma}

\noindent Note that \ref{item:2-bubble-Delta=4} differs from \ref{item:2-bubble} in that the bounding region of $H$ is unbounded from above, and in fact no such bound independent of the size of $H$ is possible. 
This is why the case of $\Delta=4$ needs to be dealt with separately. 

\begin{proof}
The proof, like for Lemma \ref{lem:Delta>=5}, constructs the required drawing by induction on the size of $H$. 
That is, to prove either of \ref{item:1-bubble-Delta=4} and \ref{item:2-bubble-Delta=4} for a bubble $H$, we assume that the entire lemma holds for any bubble with fewer vertices than $H$ has. 
We proceed along the same lines as in the proof of Lemma \ref{lem:Delta>=5}, focusing only on those details in which the two proofs differ. 

\begin{enumeratedlemmaproofitem}{item:1-bubble-Delta=4}
We consider the same cases as in the proof of \ref{item:1-bubble}. 

\begin{case}
\textit{$\ell>r$.}

As in the proof of \ref{item:1-bubble}, the statement is trivial in this case.
\end{case}

\begin{case}
\textit{$H$ is a v-bubble and\/ $1\le\ell\le r\le 2$.}

\begin{subcase}
\textit{$1\le\ell=r\le 2$.}

This is handled the same way as in Subcase \ref{subcase:v-bubble-simple-l=r} of the proof of \ref{item:1-bubble}.
\end{subcase}

\begin{subcase}
\textit{$\ell=1$ and\/ $r=2$.}

Define $H'=H-v$, and let $v^1$ and $v^2$ be the left and right neighbors of $v$ in $H$, respectively. 
As in Subcase \ref{subcase:v-bubble-simple-l<r} of the proof of \ref{item:1-bubble}, $H'$ is a $2$-bubble with roots $v^1$ and $v^2$. 
Draw $H'$ inside $\BR{v^1v^2}{1}{2}$ using the induction hypothesis \ref{item:2-bubble-Delta=4}. 
This bounding region is clearly contained in $\BR{v}{1}{2}$. 
\end{subcase}
\end{case}

\begin{case}
\textit{$H$ is a v-bubble, $0\le\ell\le r\le 3$, and\/ $\ell=0$ or\/ $r=3$.}

As $d_H(v)\le 3$, the cases $\ell=0$ and $r=3$ cannot hold simultaneously. 
Therefore, by symmetry, it is enough to consider only the case that $1\le\ell\le r=3$. 
Consider two subcases of the latter. 

\begin{subcase}
\textit{$\ell=r=3$.}

This is handled the same way as in Subcase \ref{subcase:v-bubble-complex-l=r} of the proof of \ref{item:1-bubble}.
\end{subcase}

\begin{subcase}\label{subcase:Delta=4}
\textit{$1\le\ell<r=3$.}

Define path $w_0\ldots w_n$, bubble $H'$, index $r_X$, and vertices $w_i^j$ like in Subcase \ref{subcase:v-bubble-complex-l<r} of the proof of \ref{item:1-bubble}. 
For $i=0,\ldots,n-1$, define
\begin{equation*}
\lambda_i=\begin{cases}
1&\text{if $w_i^2=w_{i+1}^1$},\\
1+\epsilon&\text{if $w_i^2\ne w_{i+1}^1$},
\end{cases}
\end{equation*}
for a small $\epsilon>0$. 
Draw vertices $w_i$ and $w_i^j$ like in Subcase \ref{subcase:v-bubble-complex-l<r} of the proof of \ref{item:1-bubble} but with the new definition of $\lambda_i$. 
All $w_i^j$ lie on a common horizontal line, and moreover the segments $w_i^2w_{i+1}^1$ (for $w_i^2\ne w_{i+1}^1$) have length $\epsilon$. 
Define $B_X=\BR{w_0w_n}{\ell}{r_X}$. 
To obtain the required drawing of $H$, it suffices to draw $H'$ inside $B_X$, and then the remaining part of $H$ can be drawn like in Subcase \ref{subcase:v-bubble-complex-l<r} of the proof of \ref{item:1-bubble}. 
But here the drawing of $H'$ inside $B_X$ is more tricky. 

If $H'$ is a $1$-bubble, then $n=1$, $w_0^2=w_1^1$, and thus we can draw $H'$ inside $\BR{w_0^2}{1}{2}$ by the induction hypothesis \ref{item:1-bubble-Delta=4}, scaling the drawing appropriately to fit it within $B_X$. 
Thus we assume that $H'$ is a $k$-bubble with $k\ge 2$. 
We split $H'$ into $2$-bubbles $X_1,\ldots,X_{k-1}$ so that the splitting sequences of $X_1,\ldots,X_{k-1}$ together form the splitting sequence of $H'$. 
The roots of $X_1,\ldots,X_{k-1}$ are pairs of vertices consecutive in the sequence of all $w_i^j$. 
Define $\ell_i=4-d_{X_i}(u_i)$ and $r_i=d_{X_i}(v_i)$, where $u_i$ and $v_i$ denote respectively the first and the last root of $X_i$. 
We draw each $X_i$ inside $\BR{u_iv_i}{\ell_i}{r_i}$ using the induction hypothesis \ref{item:2-bubble-Delta=4}. 
Since each root of $H'$ has degree at most $3$ in $H'$, we have $\ell_i\ge 2$ for $2\le i\le k-1$ and $r_i\le 1$ for $1\le i\le k-2$. 
Thus the bounding regions for $X_i$ do not overlap. 
Moreover, for $2\le i\le k-2$, we clearly have $\BR{u_iv_i}{\ell_i}{r_i}\subset B_X$. 
Thus to complete the proof for the case of $H$ being a v-bubble, it remains to show that the drawings of $X_1$ and $X_{k-1}$ are contained in $B_X$. 
We do not necessarily have $\BR{u_1v_1}{\ell_1}{r_1}\subset B_X$. 
However, this inclusion may not hold only if $\ell=2$ and $w_0^2\ne w_1^1$. 
In this case we have $\abs{u_1v_1}=\epsilon$ and thus the drawing of $X_1$ indeed lies within $B_X$ provided that $\epsilon$ is small enough. 
Similarly, the drawing of $X_{k-1}$ in contained in $B_X$ for $\epsilon$ small enough. 
\end{subcase}
\end{case}

\begin{case}
\textit{$H$ consists of at least two v-bubbles.}

This is handled the same way as in Case \ref{case:1-bubble-ind} of the proof of \ref{item:1-bubble}.
\end{case}
\end{enumeratedlemmaproofitem}

\begin{enumeratedlemmaproofitem}{item:2-bubble-Delta=4}
Again, we consider the same cases as in the proof of \ref{item:2-bubble}. 

\begin{case}
\textit{The splitting sequence of\/ $H$ contains a trivial single-edge e-bubble.}

This is handled the same way as in Case \ref{case:bridge} of the proof of \ref{item:2-bubble}. 
The only difference is that here the bounding regions of $X$ and $Y$ (defined as in Case 1 of the proof of \ref{item:2-bubble}) obtained from the application of the induction hypothesis \ref{item:2-bubble-Delta=4} are not bounded from above, so we need to choose $\lambda$ small enough for the drawings of $X$ and $Y$ not to overlap and to fit within $\BR{uv}{\ell}{r}$. 
\end{case}

\begin{case}
\textit{$H$ is a non-trivial e-bubble.}

It follows that $d_H(u),d_H(v)\ge 2$ and therefore $\ell\le 2$ and $r\ge 1$. 
Define vertices $u^j$, $v^j$ and bubble $H'$ like in Case \ref{case:e-bubble} of the proof of \ref{item:2-bubble}. 
Define 
\begin{equation*}
h=\begin{cases}
\abs{uv}&\text{if $u^2=v^1$},\\
\abs{uv}-\epsilon&\text{if $u^2\ne v^1$},
\end{cases}
\end{equation*}
for a small $\epsilon>0$. 
Put each vertex $u^i$ at point $u+h\ff_i$ and each vertex $v^i$ at point $v+h\ff_i$, so that if $u^2$ and $v^1$ are the same vertex, then they correctly end up at the same point. 
All $u^i$ and $v^i$ lie on a common horizontal line, and moreover the segment $u^2v^1$ (if exists) has length $\epsilon$. 
The same argument as in Subcase \ref{subcase:Delta=4} above shows that $H'$ can be drawn inside $\BR{u^\ell v^r}{1}{2}$. 
\end{case}

\begin{case}
\textit{The splitting sequence of\/ $H$ contains at least two e-bubbles but no trivial e-bubbles.}

Let $(H_1,\ldots,H_b)$ be the splitting sequence of $H$ and $w_0\ldots w_n$ be the root-path of $H$. 
Thus $b\ge 2$, $w_0=u$, $w_n=v$, and $n\ge 2$. 
Since none of the edges $w_0w_1,\ldots,w_{n-1}w_n$ is a bridge in $H$. 
We split $H$ into $2$-bubbles $X_1,\ldots,X_n$ so that the roots of $X_i$ are $w_{i-1}$ and $w_i$ and the splitting sequences of $X_1,\ldots,X_n$ together form the splitting sequence $(H_1,\ldots,H_b)$ of $H$. 
Define $\ell_i=4-d_{X_i}(w_{i-1})$ and $r_i=d_{X_i}(w_i)$. 
We draw each $X_i$ inside $\BR{w_{i-1}w_i}{\ell_i}{r_i}$ using the induction hypothesis \ref{item:2-bubble-Delta=4}. 
Since none of $X_1,\ldots,X_n$ is a trivial e-bubble, we have $\ell_1=\ell\le 2$, $\ell_i=2$ for $2\le i\le n$, $r_i=1$ for $1\le i\le n-1$, and $r_n=r\ge 1$. 
Thus the bounding regions for $X_i$ do not overlap and are contained in $\BR{w_0w_n}{\ell}{r}$. 
\end{case}

\begin{case}
\textit{The splitting sequence of\/ $H$ starts or ends with a v-bubble.}

This is handled the same way as in Case \ref{case:2-bubble-ind} of the proof of \ref{item:2-bubble}.\qedhere
\end{case}
\end{enumeratedlemmaproofitem}
\end{proof}

Now, to prove the Main Theorem, pick any vertex $v$ of $G$ of degree at most $\Delta-1$ (such a vertex always exists in an outerplanar graph), fix its position in the plane, and apply \ref{item:1-bubble} or \ref{item:1-bubble-Delta=4} to the graph $G$ considered as a $1$-bubble with root $v$.

\section*{Acknowledgments}

We thank V{\'{\i}}t Jel{\'{\i}}nek and D{\"o}m{\"o}t{\"o}r P{\'a}lv{\"o}lgyi for introducing us to the problem at the meeting in Prague in the summer of 2011. 
We also thank anonymous reviewers whose comments significantly contributed to improving the quality of the paper.

\bibliography{slopebib}
\bibliographystyle{plain}

\end{document}